\documentclass[conference,compsoc]{IEEEtran}
\ifCLASSOPTIONcompsoc
  \usepackage[nocompress]{cite}
\else
  \usepackage{cite}
\fi
\usepackage{xcolor}
\usepackage{multirow}
\usepackage{array}
\usepackage{booktabs}
\usepackage{makecell}
\usepackage{amsmath,amssymb,amsthm}
\usepackage{graphicx}
\usepackage{algorithm}
\usepackage{algpseudocode}
\usepackage{cuted}
\usepackage{rotating}
\usepackage{diagbox}
\usepackage{hyperref}
\usepackage{caption}

\theoremstyle{plain}
\newtheorem{theorem}{Theorem}
\newtheorem{Lemma}{Lemma}
\theoremstyle{definition}
\newtheorem{definition}{Definition}
\theoremstyle{remark}

\ifCLASSINFOpdf
\else
\fi
\begin{document}
%
\title{Quantum Cryptanalysis on Two Feistel-Type Constructions}


\author{
\IEEEauthorblockN{
Yan-Ying Zhu\IEEEauthorrefmark{1},
Bin-Bin Cai\IEEEauthorrefmark{1}\IEEEauthorrefmark{2}\IEEEauthorrefmark{3},
Fei Gao\IEEEauthorrefmark{2},
Song Lin\IEEEauthorrefmark{1}
}
\IEEEauthorblockA{
\IEEEauthorrefmark{1}College of Computer and Cyber Security, Fujian Normal University\\
Fuzhou, 350117, China\\
\IEEEauthorrefmark{2}State Key Laboratory of Networking and Switching Technology, Beijing University of Posts and Telecommunications\\
Beijing, 100876, China\\
\IEEEauthorrefmark{3}Digital Fujian Internet-of-Things Laboratory of Environmental Monitoring, Fujian Normal University\\
Fuzhou, 350117, China\\
Corresponding authors: Bin-Bin Cai (cbb@fjnu.edu.cn), Song Lin (lins95@fjnu.edu.cn)
}
}


%


\maketitle

\begin{abstract}
With the rapid advancement of quantum computing, the security of traditional cryptographic systems is facing significant challenges. As a fundamental component of symmetric cryptography, block ciphers have attracted increasing attention regarding their quantum security in recent years.

In this paper, we propose new quantum attacks on two Feistel-type constructions, i.e., New Structure IV and the lightweight block cipher LBlock, under the quantum chosen-plaintext attack (qCPA) model.
First, we construct an 8-round quantum distinguisher for New Structure IV and perform an $r(r>12)$-round quantum key-recovery attack. Compared with Li et al.'s qCPA-based result presented at INSCRYPT 2024, our attack reduces the time complexity by a factor of $2^{n/2}$, where $n$ denotes the subkey length.
Second, we design a 9-round quantum distinguisher for LBlock and extend it to an $r(r>15)$-round quantum key-recovery attack. The proposed quantum distinguisher covers one more round than Liu et al.'s automated-search distinguisher at ASIACRYPT 2025, which results in a $2^{16}$ reduction in the time complexity of our $r$-round quantum key-recovery attack on LBlock. Furthermore, by analyzing the structure of the round function and the key scheduling of LBlock, we reduce the subkey guess space for the 15-round attack from 76 bits to 52 bits, lowering the time complexity from \(O(2^{38})\) to \(O(2^{26})\) by a factor of \(2^{12}\).
\end{abstract}


%
\IEEEpeerreviewmaketitle

\section{Introduction}

Block ciphers are fundamental symmetric primitives that use a single key to encrypt and decrypt fixed-length data blocks. They offer high efficiency and wide applicability, making them a cornerstone of modern cryptographic systems. Over the past few decades, a variety of classical cryptanalytic techniques, including differential \cite{Biham1990}, linear \cite{Matsui1993}, and impossible differential cryptanalysis \cite{Biham1999}, have been proposed to evaluate their security. However, with the rapid advancement of quantum computing \cite{CG2022,CW2023,SJ2025,Zhen2026,ZC2026}, adversaries can exploit quantum parallelism and quantum interference to significantly accelerate cryptanalysis, enabling quantum attacks to outperform classical attacks in terms of time complexity. This motivates a rigorous reassessment of the quantum security of existing block ciphers.

In the quantum setting, quantum algorithms provide new tools for analyzing the security of block ciphers \cite{Grover1996,Simon1997,Leander2017,Bonnetain2018,Zhou2021,Zhou2023}. For example, Grover algorithm \cite{Grover1996} offers a quadratic speedup for exhaustive key search, reducing the time complexity from $O(2^n)$ to $O(2^{n/2})$, where $n$ denotes the key length. Simon algorithm \cite{Simon1997} can efficiently recover the hidden periodic structure of a function, which can be exploited to construct quantum distinguishers that distinguish a target cipher from a random permutation. Combining the strengths of both algorithms, the Grover-meets-Simon framework \cite{Leander2017} was proposed. This framework uses Grover algorithm to search over candidate subkeys and, for each guessed subkey, invokes Simon algorithm to verify whether the corresponding function has a hidden period. These algorithms have numerous extensions in cryptanalysis, including Grover meets Kuperberg algorithm\cite{Bonnetain2018}, Bernstein–Vazirani meet Grover algorithm\cite{Zhou2021}, and the distributed Simon algorithm\cite{Zhou2023}. Although many quantum algorithms offer significant cryptanalytic potential, effectively applying them to concrete cipher structures remains challenging.

To evaluate the practical impact of these quantum algorithms on symmetric cryptography, extensive studies have been conducted on the quantum security of block ciphers. Among various block cipher structures, generalized Feistel structures have attracted considerable attention due to their structural regularity. Dong et al. \cite{Dong2019} first analyzed the quantum security of Type-1 and Type-2 generalized Feistel structures (GFSs). They constructed polynomial-time quantum distinguishers under the quantum chosen-plaintext attack (qCPA) model and extended them to quantum key-recovery attacks by appending additional rounds. Following this line of research, Ni et al. \cite{Ni2019} and Sun et al. \cite{Sun2023} optimized distinguisher constructions and input selection strategies, thereby extending the quantum distinguishing attack against Type-1 GFS to \(d^2-1\) rounds, where $d$ denotes the number of branches. Polynomial-time quantum distinguishing attacks were also proposed for Type-3 GFS \cite{S2020,Zhang2023}. Beyond the conventional GFSs, researchers also explored other generalized Feistel variants, including New Structure IV \cite{Wu2011}. Recently, Li and Wei studied New Structure IV in the quantum chosen-ciphertext attack (qCCA) model and obtained new cryptanalytic results \cite{Li2024,Wei2025}. However, under the qCPA model, extending quantum distinguishers for New Structure IV to more rounds remains an open challenge.

Moreover, lightweight block ciphers are widely used in the Internet of Things and embedded devices due to their compact and efficient design. LBlock \cite{LBlock2011}, as a representative lightweight block cipher, features a distinctive round function and key scheduling. Although LBlock has been analyzed under quantum attack models~\cite{Xiang2024,Liu2025}, developing more effective quantum cryptanalytic techniques remains an open challenge.

\begin{table}[!t]
\centering
\caption{Comparison of quantum attacks on New Structure IV and LBlock under the qCPA model.}
\label{table1}
\footnotesize
\renewcommand{\arraystretch}{1.2}
\setlength{\tabcolsep}{3pt}

\begin{tabular}{
>{\centering\arraybackslash}m{1.5cm}
>{\centering\arraybackslash}m{0.9cm}
>{\centering\arraybackslash}m{1.7cm}
>{\centering\arraybackslash}m{2cm}
>{\centering\arraybackslash}m{1.2cm}
}
\toprule
\textbf{Target} &
\textbf{Round} &
\textbf{Key-Recovery Rounds} &
\textbf{Complexity} &
\textbf{Reference} \\
\midrule

\multirow{4}{*}{\makecell[c]{New\\Structure IV}}
& 5 & $r > 9$
& $\frac{2n+(r-9)n}{2}$
& \cite{Bo2021} \\
& 6 & $r > 10$
& $\frac{2n+(r-10)n}{2}$
& \cite{Kun2022} \\
& 7 & $r > 11$
& $\frac{3n+(r-11)n}{2}$
& \cite{Li2024} \\

& 8 & $r > 12$
& $\frac{3n+(r-12)n}{2}$
& Section 3 \\

\midrule

\multirow{3}{*}{LBlock}
& 8 & $r > 14$
& $38+(r-14)\times16$
& \cite{Xiang2024} \\

& 8 & $r > 14$
& $38+(r-14)\times16$
& \cite{Liu2025} \\

& 9 & $r > 15$
& $38+(r-15)\times16$
& Section 4 \\

\bottomrule
\end{tabular}
\end{table}

To solve these two challenges, this paper focuses on the security analysis of New Structure IV and LBlock in the quantum setting. Specifically, we construct an 8-round quantum distinguisher for New Structure IV under the qCPA model. Based on this, we present a quantum key-recovery attack on the \(r(r>12)\)-round New Structure IV. The attack achieves a time complexity of \(O(2^{((r-12)n+3n)/2})\), where $n$ denotes the subkey length. This reduces the time complexity by a factor of \(2^{n/2}\) compared with the previous best-known quantum key-recovery attack for this structure under the qCPA model.
We further extend our attack framework to the LBlock cipher as a concrete instantiation. For LBlock, we design a 9-round quantum distinguisher and perform a general quantum key-recovery attack on its \(r(r>15)\)-round version, with a time complexity of \(O(2^{38+(r-15)\times 16})\). Table~\ref{table1} provides a comparison of this work with existing quantum attacks on New Structure IV and LBlock. Our results outperform previous works in terms of both the number of rounds covered by the distinguisher and the overall time complexity of the attacks. Furthermore, we conduct a thorough analysis of the round function and the key scheduling of LBlock. Based on its structural properties, we reduce the subkey guessing space from 76 bits to 52 bits in the 15-round quantum key-recovery attack, thereby lowering the time complexity to \(O(2^{26})\).In the practical attack scenario against 15-round LBlock, the proposed attack reduces the complexity by a factor of $2^{14}$ compared with a Grover-based exhaustive search.

\section{Preliminaries}
In this section, we present the necessary preliminaries, including fundamental definitions and relevant quantum algorithms.

\subsection{Notation}
Let \( E : \{0,1\}^n \times \{0,1\}^m \to \{0,1\}^n \) be a block cipher with block size $n$ and key size $m$. For a given key \( k \in \{0,1\}^m \), we denote by \( E_k(x) \) the encryption of a plaintext \( x \in \{0,1\}^n \).

Let \( f:\{0,1\}^n \to \{0,1\}^n \) be a function. In the quantum setting, \( f \) can be implemented as a unitary operator \( U_f \), defined as
\[
U_f : \sum_{x,y} \alpha_{x,y} \lvert x \rangle \lvert y \rangle \;\longmapsto\; \sum_{x,y} \alpha_{x,y} \lvert x \rangle \lvert y \oplus f(x) \rangle,
\]
where \( \alpha_{x,y}\in\mathbb C \) satisfy
$
\sum_{x,y}|\alpha_{x,y}|^2=1.
$

For an \(r\)-round, \(d\)-branch iterated structure, \( x_j^i \) represents the output of the ($j + 1$)-th branch after the $i$-th round, where \(0 \le j \le d-1\) and \(1\le i\le r\).

\subsection{Grover algorithm}
Grover algorithm is a quantum search algorithm for finding a unique marked element in an unstructured dataset. It iteratively applies a phase oracle and a diffusion operator to amplify the probability amplitude of the target state, thereby achieving a quadratic speedup from \( O(N) \) to \( O(\sqrt{N}) \), where $N$ denotes the size of the search space.

\begin{definition}
(Grover \cite{Grover1996}) Consider a search space of size $N=2^n$. Let $f: \{0,1\}^n \to \{0,1\}$ be a Boolean function such that there exists a unique $x' \in \{0,1\}^n$ satisfying $f(x')=1$, while $f(x)=0$ for all $x\neq x'$. The goal is to identify $x'$.
\end{definition}

Assume access to a quantum oracle $\mathcal{O}$ implementing $f$ as a phase oracle, i.e.,
$\mathcal{O} \lvert x \rangle = (-1)^{f(x)} \lvert x \rangle$.  The quantum circuit of Grover algorithm is shown in Fig. \ref{Grover}.
\begin{figure}
    \centering
    \includegraphics[width=1\linewidth]{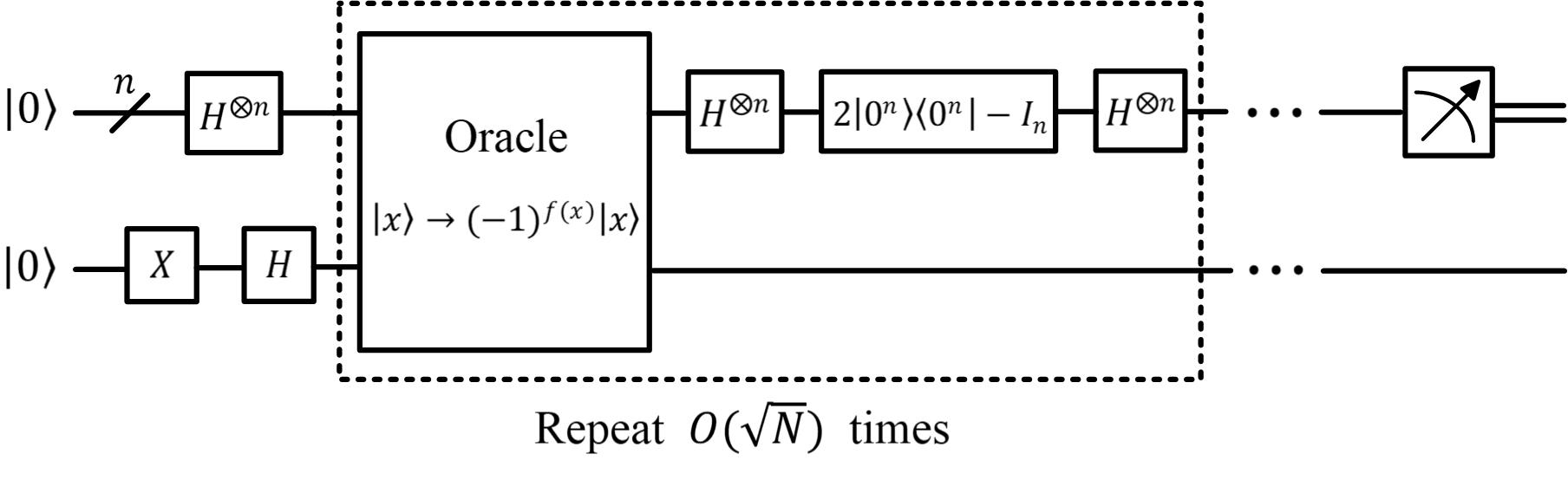}
    \caption{Quantum circuit of Grover algorithm.}
    \label{Grover}
\end{figure}
The steps of Grover algorithm are summarized in Algorithm~\ref{algorithm1}.

\begin{algorithm}[htbp]
\caption{The process of Grover algorithm}
\algrenewcommand\alglinenumber[1]{\normalsize #1.}
\begin{algorithmic}[1]
\State Prepare the initial quantum state $| 0^{\otimes n} \rangle$.
\State  Perform a Hadamard transform $H^{\otimes n} $ on the register to obtain
\[
|\psi\rangle=H^{\otimes n}| 0^{\otimes n} \rangle=\frac{1}{\sqrt{2^{n}}} \sum_{x \in\{0,1\}^{n}}|x\rangle.
\]
\State Apply Grover iteration for $R (R\approx\frac{\pi}{4} \sqrt{2^{n}})$ times to yield
\[
[(2|\psi\rangle\langle\psi|-I) \mathcal{O}]^{R}|\psi\rangle \approx\left|x'\right\rangle.
\]
Here, $\mathcal{O}$ is the phase oracle, and $(2|\psi\rangle\langle\psi|-I)$ represents the diffusion operator.
\State Measure the $n$-qubit to get $x'$.
\end{algorithmic}
\label{algorithm1}
\end{algorithm}

Furthermore, Quantum Amplitude Amplification (QAA), a generalization of Grover algorithm, is given in the following theorem.

\begin{theorem}\textnormal{(Brassard \cite{Brassard2002})}
Let $\mathcal{A}$ be any quantum algorithm acting on $q$ qubits without intermediate measurements. Let $\mathcal{B}: \{0,1\}^q \to \{0,1\}$ be a Boolean function that classifies the outcomes of $\mathcal{A}$ as good or bad states. Let $p>0$ denote the initial success probability of obtaining a good state when measuring $\mathcal{A}|0\rangle$.
Set $k = \left\lceil \frac{\pi}{4\theta} \right\rceil$, where $\theta$ is defined as $\sin^2(\theta) = p$.
Define the unitary operator
\[
Q = -\mathcal{A}S_0\mathcal{A}^{-1}S_{\mathcal{B}},
\]
where $S_{\mathcal{B}}$ flips the phase of all good states, i.e.,
\[
|x\rangle \mapsto
\begin{cases}
-|x\rangle, & \text{if } \mathcal{B}(x)=1 \\
|x\rangle, & \text{otherwise}
\end{cases},
\tag{3}
\]
and $S_0$ changes the phase of the zero state $\lvert 0\rangle$.
After applying $Q^k$ to $\mathcal{A}|0\rangle$, a measurement yields a good state with probability at least $\max\{1-p, p\}$.
\end{theorem}

\subsection{Simon algorithm}

Simon algorithm is a fundamental quantum algorithm for identifying a hidden XOR period of a two-to-one function. It achieves an exponential speedup over classical algorithms.

\begin{definition}[Simon \cite{Simon1997}]
Let $f:\{0,1\}^n \rightarrow \{0,1\}^n$ be a function. Assume that there exists a unique non-zero period $s \in \{0,1\}^n$ satisfying
\[
f(x) = f(y) \iff x \oplus y \in \{0^n, s\}, \quad \forall x,y \in \{0,1\}^n.
\]
The goal is to find the hidden period $s$.
\end{definition}

\begin{algorithm}[htbp]
\caption{The process of Simon algorithm}
\algrenewcommand\alglinenumber[1]{\normalsize #1.}
\begin{algorithmic}[1]
\State Prepare two $n$-qubit registers initialized to $|0\rangle^{\otimes n}|0\rangle^{\otimes n}$.
\State Apply the Hadamard transform $H^{\otimes n}$ to the first register to obtain
\[
|\psi_{1}\rangle=\frac{1}{\sqrt{2^{n}}} \sum_{x \in \{0,1\}^{n}} |x\rangle |0^n\rangle.
\]
\State Query the oracle $U_f$, which maps the state to
\[
|\psi_{2}\rangle=\frac{1}{\sqrt{2^{n}}} \sum_{x \in \{0,1\}^{n}} |x\rangle |f(x)\rangle.
\]
\State Measure the second register to obtain $f(x_0)$. The first register collapses to
\[
|\psi_{3}\rangle=\frac{1}{\sqrt{2}}\left(|x_{0}\rangle + |x_{0} \oplus s\rangle\right).
\]
\State Apply the Hadamard transform $H^{\otimes n}$ to the first register again to obtain
\[
|\psi_{4}\rangle=\frac{1}{\sqrt{2^{n+1}}} \sum_{y \in \{0,1\}^{n}}
\left[(-1)^{y \cdot x_{0}} + (-1)^{y \cdot (x_{0} \oplus s)}\right] |y\rangle,
\]
where $(-1)^{y \cdot x_{0}}+(-1)^{y \cdot\left(x_{0} \oplus s\right)}=(-1)^{y \cdot x_{0}}(1+(-1)^{y \cdot s})$. The coefficient is non-zero if and only if $y \cdot s=0$.
\State Measure the first register to obtain a string $y$.
\end{algorithmic}
\label{alg:simon}
\end{algorithm}

The best known classical algorithm for solving this problem requires $O(2^{n/2})$ time. In contrast, Simon algorithm recovers the hidden period $s$ with only $O(n)$ quantum queries. The quantum circuit of Simon algorithm is shown in Fig. \ref{Simon}, and the corresponding procedure is summarized in Algorithm~\ref{alg:simon}.
\begin{figure}
    \centering
    \includegraphics[width=0.9\linewidth]{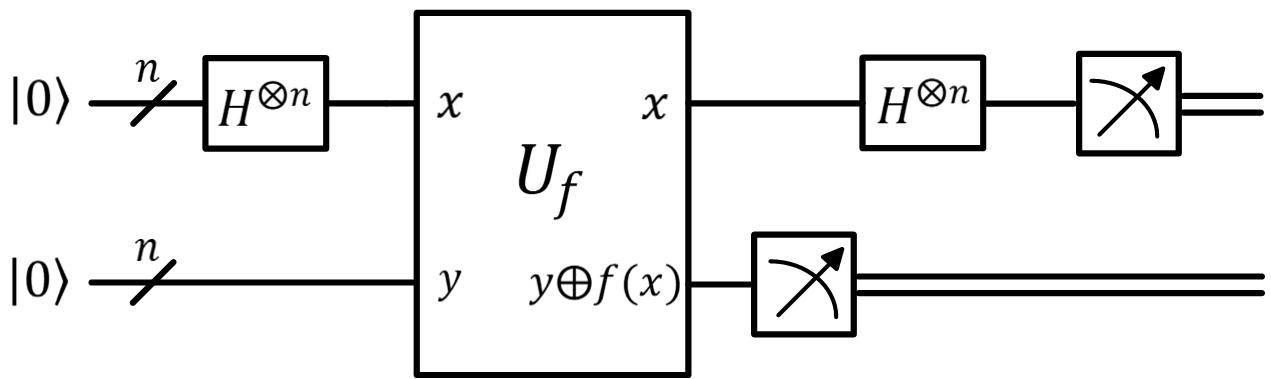}
    \caption{Quantum circuit of Simon algorithm.}
    \label{Simon}
\end{figure}
Repeating the procedure $O(n)$ times yields $(n-1)$ linearly independent equations of the form $y \cdot s = 0$, from which $s$ can be recovered via Gaussian elimination.

However, in practical cryptographic constructions, a function derived from a cipher may exhibit additional collisions and thus fail to strictly satisfy the two-to-one promise of the Simon problem. To address this issue, Kaplan et al.~\cite{Kaplan2016} introduced a relaxed version of Simon problem by defining the parameter
\[
\varepsilon(f, s)=\max_{t \in \{0,1\}^n \backslash \{0^n, s\}} \operatorname{Pr}_x[f(x) = f(x \oplus t)],
\]
which quantifies how far the function deviates from satisfying Simon promise. They proved that if there exists a constant $p_0<1$ such that $\varepsilon(f,s) \le p_0$, then the Simon algorithm can recover the period $s$ using $cn$ oracle queries with probability at least
$
1 - ( 2 (\frac{1+p_0}{2} )^c )^n
$, where $c$ is a constant.

\subsection{Grover-meets-Simon algorithm}

At Asiacrypt 2017, Leander et al.~\cite{Leander2017} proposed the Grover-meets-Simon algorithm to attack the FX construction. They showed that key whitening does not significantly improve security in the quantum setting.

\begin{definition}[Leander~\cite{Leander2017}]
Let $f:\{0,1\}^{m} \times \{0,1\}^{n} \rightarrow \{0,1\}^{i}$ be a function, where $m=O(n)$.
Assume that there exist a unique index
$i_0\in\{0,1\}^m$
and a non-zero period
$s\in\{0,1\}^n$
such that
\[
f(i_0,x)=f(i_0,x\oplus s),
\quad \forall x\in\{0,1\}^n.
\]
The goal is to find $i_0$ and the hidden period $s$.
\end{definition}

The FX construction is defined as $\operatorname{Enc}(x)=E_{k_0}(x \oplus k_1) \oplus k_2$, as illustrated in Fig.~\ref{FX}.
\begin{figure}
    \centering
    \includegraphics[width=0.65\linewidth]{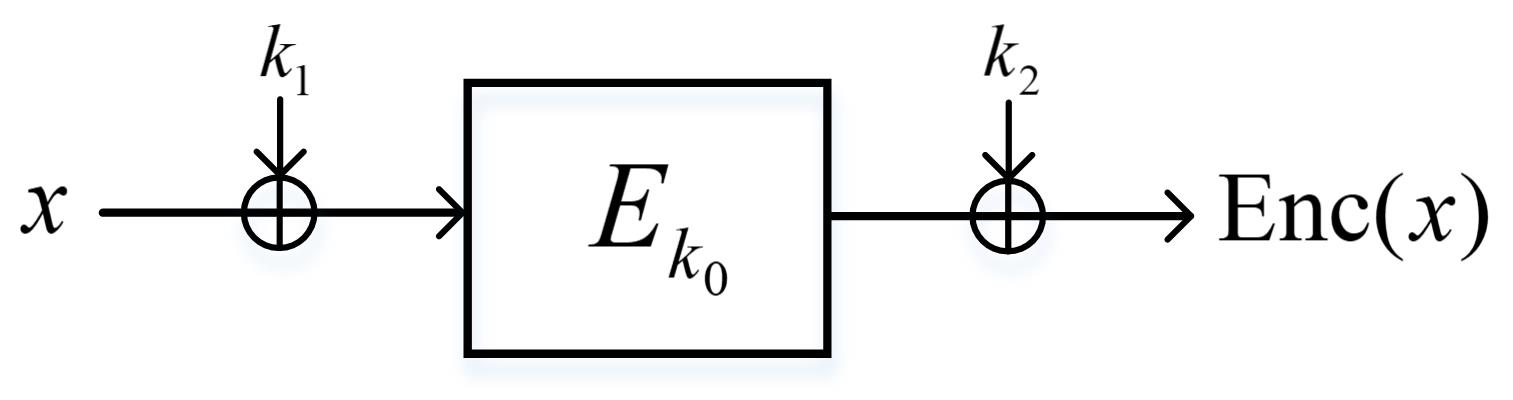}
    \caption{FX structure.}
    \label{FX}
\end{figure}
To recover the key $k_1$, Leander et al. considered the following function
\[
f(k,x) = E_{k_0}(x \oplus k_1) \oplus k_2 \oplus E_k(x).
\]
Note that when $k = k_0$, the function satisfies $f(k_0, x) = f(k_0, x \oplus k_1)$ for all $x \in \{0,1\}^n$. Therefore, \(f(k_0,\cdot)\) has the hidden period \(k_1\). In contrast, for any $k \neq k_0$, the function \(f(k,\cdot)\) does not exhibit this periodicity. The attack proceeds as follows.
\begin{enumerate}
    \item Perform an exhaustive search over the key space using Grover algorithm.
    \item For each candidate $k'$, apply Simon algorithm to determine whether the function $f(k', \cdot)$ has a hidden period. If $k' = k_0$, then $f(k', \cdot)$ is periodic with period $k_1$, and Simon algorithm can efficiently distinguish it. Otherwise, no such period exists, and the candidate $k'$ is rejected.
    \item Once $k_0$ and $k_1$ have been recovered, $k_2$ can be easily computed.
\end{enumerate}

\section{Quantum attack on New Structure IV}\label{sec3}

In this section, we study the quantum security of New Structure IV. Under the qCPA model, we construct an 8-round quantum distinguisher and propose a quantum key-recovery attack against $r(r>12)$-round New Structure IV.
\par \subsection{A Brief Description of New Structure IV}
Wu et al. \cite{Wu2011} proposed four types of generalized unbalanced Feistel structures, namely New Structure I--IV. Among them, we focus on New Structure IV, whose structure is shown in Fig. \ref{iNew}. The block size of New Structure IV is \(4n\) bits, and the input is divided into four $n$-bit branches. In each round, $F^i$ denotes the function with subkey $k^i$.
\begin{figure}
    \centering
    \includegraphics[width=0.6\linewidth]{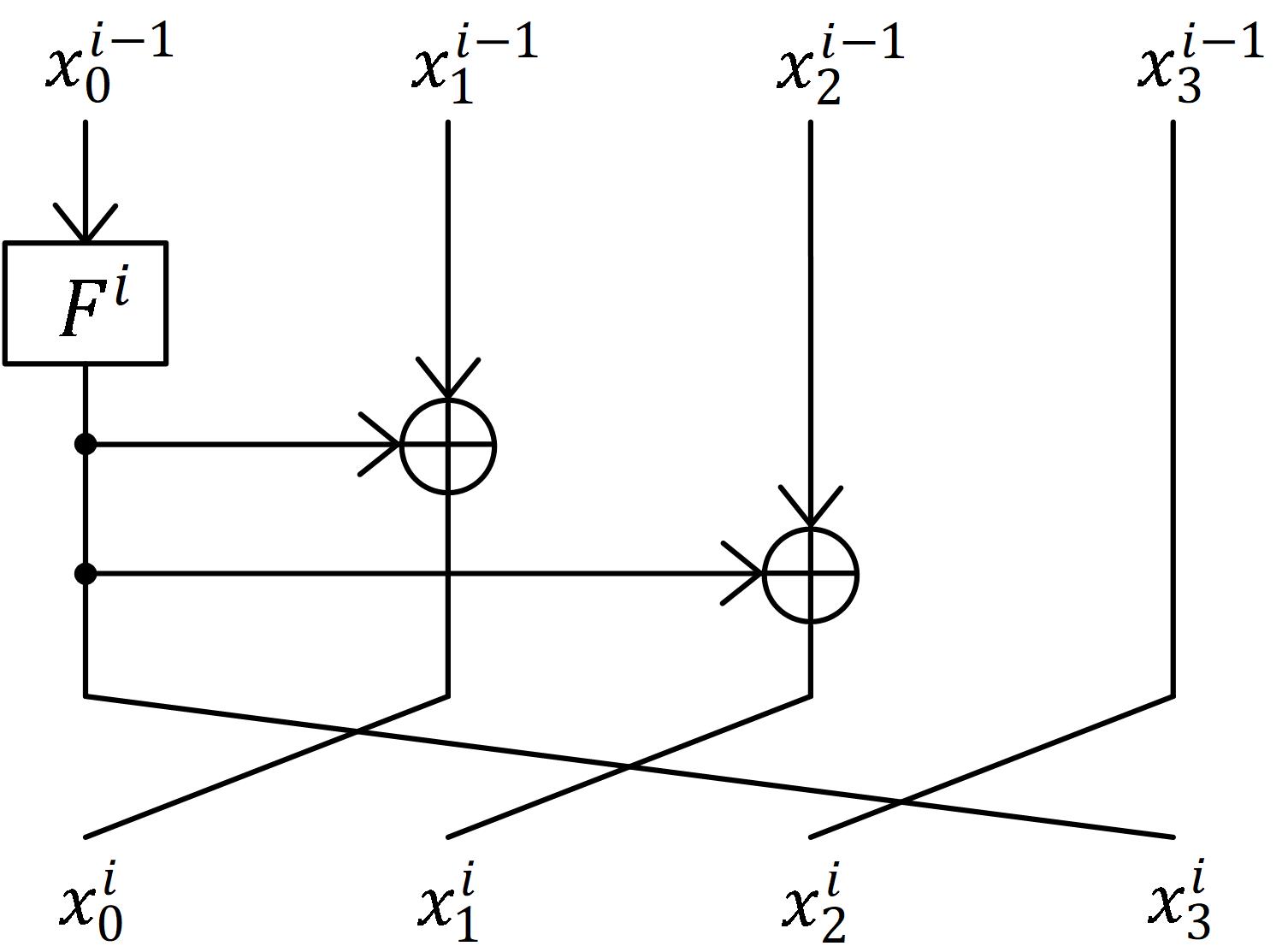}
    \caption{The $i$-th round of New Structure IV.}
    \label{iNew}
\end{figure}

Let $(x_{0}^{i-1}, x_{1}^{i-1}, x_{2}^{i-1}, x_{3}^{i-1}) \in (\{0,1\}^n)^4$ be the input to the $i$-th round of New Structure IV. The corresponding output can be expressed as
\[
\begin{aligned}
&(x_{0}^{i}, x_{1}^{i}, x_{2}^{i}, x_{3}^{i})\\
&=\big(x_{1}^{i-1} \oplus F^{i}(x_{0}^{i-1}),\; x_{2}^{i-1} \oplus F^{i}(x_{0}^{i-1}), 
 x_{3}^{i-1},\; F^{i}(x_{0}^{i-1})\big).
\end{aligned}
\]
It can be observed that the transformation is mainly determined by $F^{i}(x_0^{i-1})$, which will be exploited in our quantum distinguisher and quantum key-recovery attacks.

Under the qCPA model, Yu et al. \cite{Bo2021} first constructed a 5-round quantum distinguisher for New Structure IV. Based on this distinguisher, they further achieved a 9-round quantum key-recovery attack. Later, Zou et al. \cite{Kun2022} improved the quantum distinguishing attack and extended the distinguisher to 6 rounds. More recently, Li et al. \cite{Li2024} proposed a 7-round quantum distinguisher for New Structure IV.
In this paper, we further advance this line of research. We extend the distinguisher to 8 rounds under the qCPA model and present an $r(r>12)$-round quantum key-recovery attack on New Structure IV. Compared with the previous best-known attack, our time complexity is reduced by a factor of \(2^{n/2}\).

\subsection{8-round quantum distinguisher for New Structure IV }

For the 8-round New Structure IV (Fig. \ref{8New}), we design a periodic function \( f^{\mathcal{O}_1}(x) \) such that Simon algorithm can distinguish it from a random permutation in \( O(n) \) time, where \( n \) denotes the subkey length.

Let $\alpha_0, \alpha_1, c \in \{0,1\}^n$ be arbitrary constants with $\alpha_0 \neq \alpha_1$. Given the 8-round New Structure IV encryption oracle $\mathcal{O}_1$, consider inputs of the form $(x_0^0, x_1^0, x_2^0, x_3^0) = (c, \alpha_b, x, x)$, where $b \in \{0,1\}$. Denote the corresponding ciphertext after 8 rounds by
\[
(x_0^8, x_1^8, x_2^8, x_3^8) = \mathcal{O}_1(c, \alpha_b, x, x).
\]
For each $b \in \{0,1\}$, define
\begin{equation}
 y_1^{(\alpha_b)}=x_0^8(\alpha_b) \oplus x_2^8(\alpha_b) \oplus x_3^8(\alpha_b).
\label{Eq1}
\end{equation}
Here, \( x_i^j(\alpha_b) \) denotes the output of the \((i+1)\)-th branch after the \( j \)-th round of New Structure IV with plaintext \( (c, \alpha_b, x, x)\), where \( i \in \{0,1,2,3\} \) and $j$ denotes the round index.

Based on this, we define the function
\[
\begin{aligned}
f^{\mathcal{O}_1} : \{0,1\}^n &\rightarrow \{0,1\}^n \\
x &\mapsto y_1^{(\alpha_0)} \oplus y_1^{(\alpha_1)}.
\end{aligned}
\]
To further characterize the structural properties of this function, we present Lemma~\ref{Lemma1} as follows.

\begin{figure}
    \centering
    \includegraphics[width=1.0\linewidth]{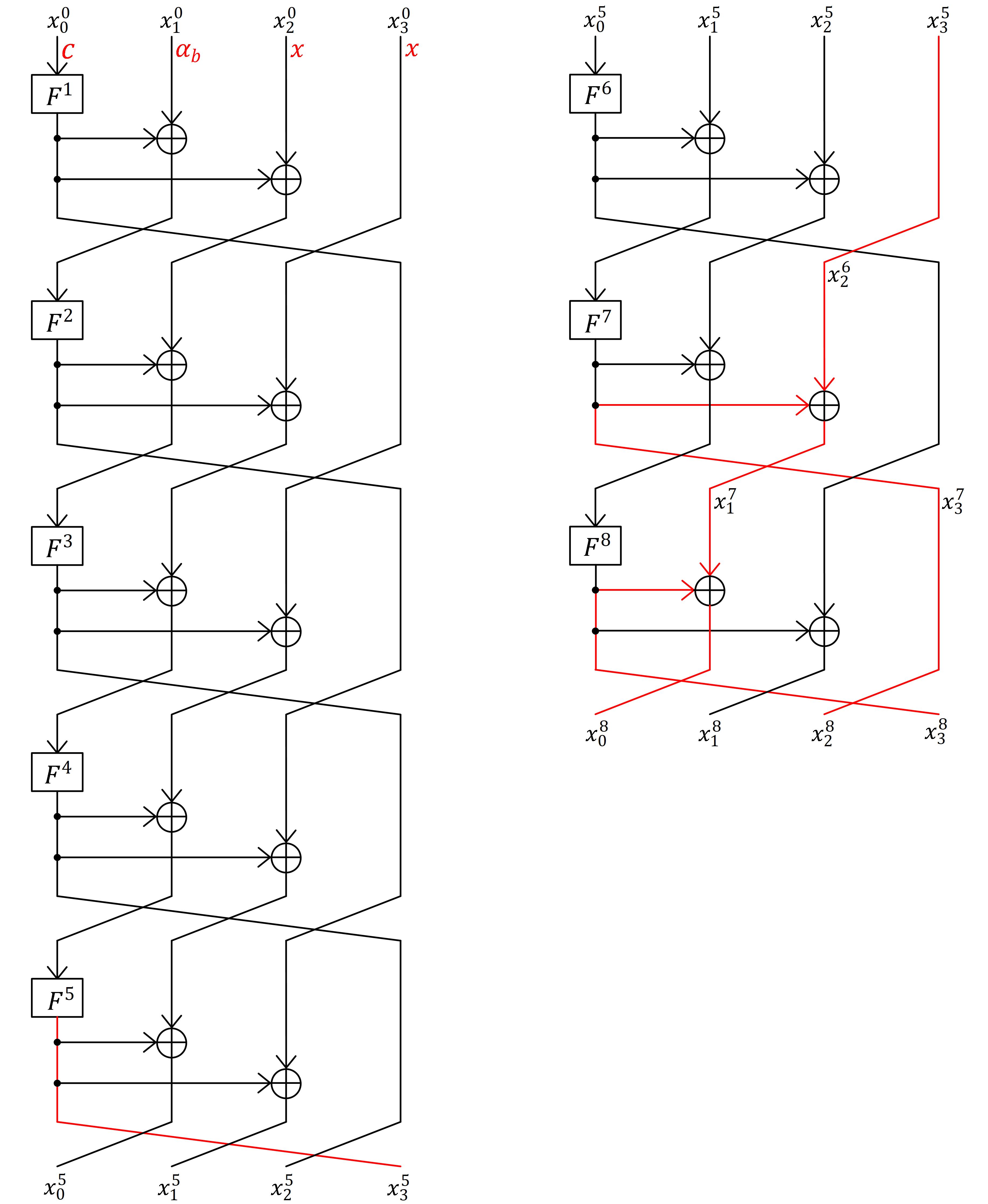}
    \caption{8-round quantum distinguisher for New Structure IV.}
    \label{8New}
\end{figure}

\begin{Lemma}\label{Lemma1}
Under the qCPA model, let $\mathcal{O}_1$ be the encryption oracle of 8-round New Structure IV. For any \(x \in \{0,1\}^n\), the function \(f^{\mathcal{O}_1}(x)\) satisfies
\[
f^{\mathcal{O}_1}(x) = f^{\mathcal{O}_1}\bigl(x \oplus F^2(\alpha_0 \oplus F^1(c)) \oplus F^2(\alpha_1 \oplus F^1(c))\bigr),
\]
where \(c\) is the fixed constant used for the initial input $x_0^0$. That is, \(f^{\mathcal{O}_1}\) has period
$
s = F^2(\alpha_0 \oplus F^1(c)) \oplus F^2(\alpha_1 \oplus F^1(c)).
$
\end{Lemma}

\begin{proof}
For each $b\in\{0,1\}$, we query the 8-round New Structure IV encryption oracle $\mathcal{O}_1$ with plaintext 
$(c,\alpha_b, $ $ x, x).$ 
From the structure in Fig.~\ref{8New}, we have
\[
x_0^8 \oplus x_2^8 \oplus x_3^8 = x_1^7 \oplus x_3^7 = x_2^6 = x_3^5.
\]
By Eq.~\eqref{Eq1}, we further obtain $y_1^{(\alpha_b)} =x_3^5(\alpha_b)$. This implies that the function \(f^{\mathcal{O}_1}(x)\) depends only on the output of the 4th branch at the 5th round. Consequently, \(f^{\mathcal{O}_1}(x)\) can be equivalently expressed as
\[
f^{\mathcal{O}_1}(x) = x_3^5(\alpha_0) \oplus x_3^5(\alpha_1).
\]

By tracing the first five rounds of $\mathcal{O}_1$, we obtain
\[
\begin{aligned}
x_3^5(\alpha_b)
&= F^5\bigl(
F^1(c) \oplus F^3(x \oplus F^1(c) \oplus F^2(\alpha_b \oplus F^1(c))) \\
&\qquad \oplus F^4(x \oplus F^2(\alpha_b \oplus F^1(c)) \\
&\qquad \oplus F^3(x \oplus F^1(c) \oplus F^2(\alpha_b \oplus F^1(c))))
\bigr).
\end{aligned}
\]
Therefore, the function \(f^{\mathcal{O}_1}(x)\) can be further written as
\[
\begin{aligned}
f&^{\mathcal{O}_1}(x)\\ &= y_1^{(\alpha_0)} \oplus y_1^{(\alpha_1)} \\
&= x_3^5(\alpha_0) \oplus x_3^5(\alpha_1) \\
&= F^5\bigl(F^1(c) \oplus F^3(x \oplus F^1(c) \oplus F^2(\alpha_0 \oplus F^1(c)))\\
&\qquad\oplus F^4(x \oplus F^2(\alpha_0 \oplus F^1(c))
  \oplus F^3(x \oplus F^1(c)\\
  &\qquad \oplus F^2(\alpha_0 \oplus F^1(c))))\bigr)
 \oplus F^5\bigl(F^1(c) \oplus F^3(x \oplus\\
 &\qquad  F^1(c) \oplus F^2(\alpha_1 \oplus F^1(c)))
  \oplus F^4(x \oplus F^2(\alpha_1 \oplus\\
  &\qquad F^1(c)) \oplus F^3(x \oplus F^1(c) \oplus F^2(\alpha_1 \oplus F^1(c))))\bigr).
\end{aligned}
\]
The function \(f^{\mathcal{O}_1}(x)\) has a period \(s\), since it satisfies
\[
\begin{aligned}
f&^{\mathcal{O}_1}\bigl(x \oplus F^2(\alpha_0 \oplus F^1(c)) \oplus F^2(\alpha_1 \oplus F^1(c))\bigr)\\
&= F^5\bigl(F^1(c) \oplus F^3(x\oplus F^2(\alpha_1 \oplus F^1(c)) \oplus F^1(c))\\
&\quad  \oplus F^4(x \oplus F^2(\alpha_1 \oplus F^1(c)) \oplus F^3(x \oplus F^2(\alpha_1 \oplus F^1(c))\\
&\quad \oplus F^1(c)) )\bigr) \oplus F^5\bigl(F^1(c) \oplus F^3(x\oplus F^2(\alpha_0 \oplus F^1(c)) \oplus \\
&\quad F^1(c) )  \oplus F^4(x \oplus F^2(\alpha_0 \oplus F^1(c)) \oplus F^3(x\oplus F^2(\alpha_0 \oplus\\
&\quad  F^1(c)) \oplus F^1(c)) )\bigr) \\
&= f^{\mathcal{O}_1}(x),
\end{aligned}
\]
where $s = F^2(\alpha_0 \oplus F^1(c)) \oplus F^2(\alpha_1 \oplus F^1(c))$. This completes the proof of Lemma~\ref{Lemma1}.
\end{proof}
Through a rigorous mathematical proof, we establish the periodicity of the function \(f^{\mathcal{O}_1}\). Therefore, under the qCPA model, we can construct an 8-round quantum distinguisher for New Structure IV. By applying Simon algorithm, the period \(s\) can be efficiently recovered in \(O(n)\) time.

\subsection{Quantum key-recovery attack on New Structure IV}
Based on the constructed 8-round quantum distinguisher, this section presents a quantum key-recovery attack on $r(r>12)$-round New Structure IV. Compared with the attack proposed by Li et al.~\cite{Li2024}, our attack reduces the query complexity by a factor of \(2^{n/2}\), where \(n\) denotes the subkey length.

We extend the 8-round distinguisher by appending four additional rounds, as shown in Fig.~\ref{12New}. The red dashed boxes denote the three \(n\)-bit subkeys of the round function that must be guessed. Therefore, the Grover search complexity is \(O(2^{3n/2})\). The attack procedure is as follows.

\begin{figure}
    \centering
    \includegraphics[width=0.5\linewidth]{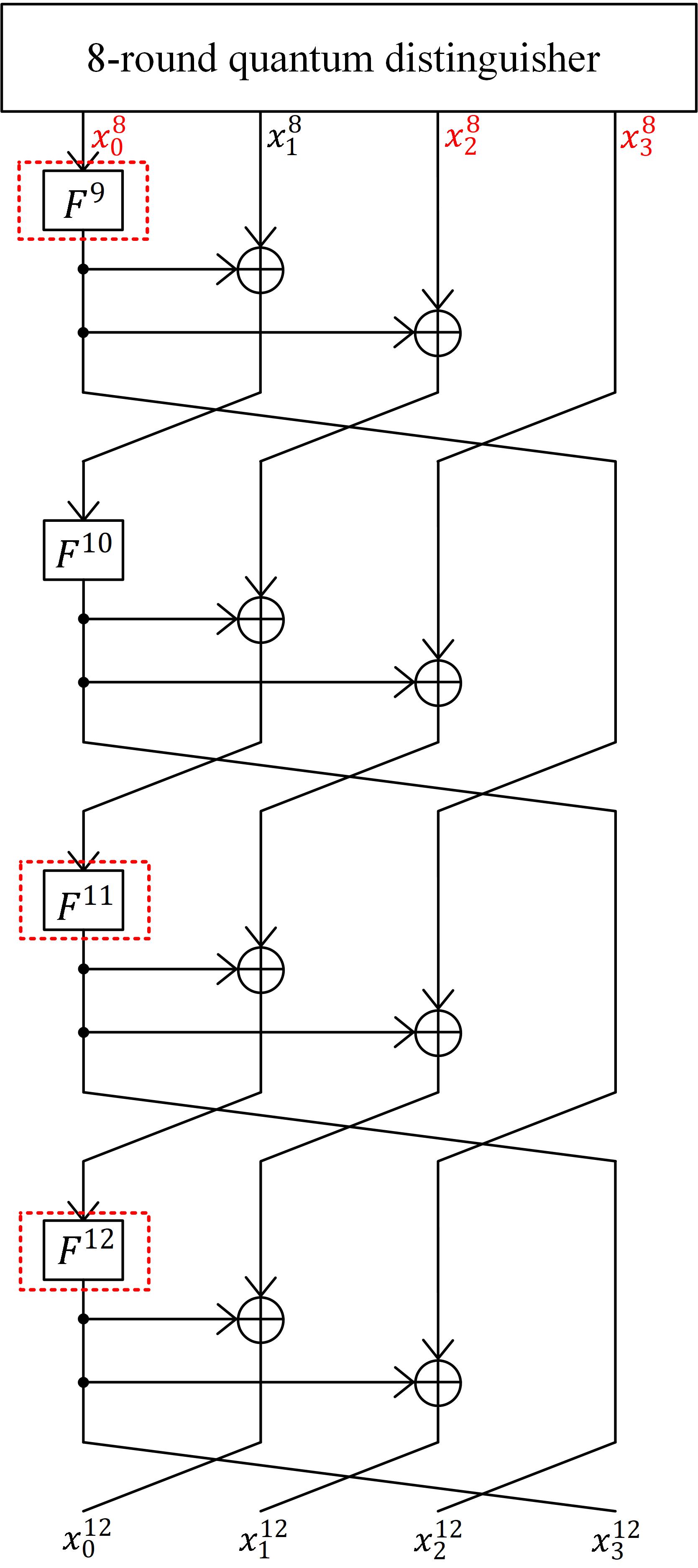}
    \caption{Quantum key-recovery attack on New Structure IV.}
    \label{12New}
\end{figure}

(1) Query the 12-round New Structure IV encryption oracle \(\mathcal{O}_2\) with the plaintext \((c,\alpha_b,x,x)\) to obtain the corresponding ciphertext \((x_0^{12}, x_1^{12}, x_2^{12}, x_3^{12})\). Then, partially decrypt the ciphertext using the subkeys of rounds 9, 11, and 12 to recover the intermediate state \((x_0^8, x_2^8, x_3^8)\).

(2) Apply Grover algorithm to search over the \(3n\)-bit subkeys $(k^9,k^{11},k^{12})$ of the last 4 rounds. For each candidate subkey tuple, we use the 8-round quantum distinguisher to verify its correctness. Only the correct subkey tuple yields the periodic function described in Lemma~\ref{Lemma1}, which can then be identified using Simon algorithm.

For the \(r(r>12)\)-round New Structure IV, Grover algorithm is used to search over a \(((r-12)n + 3n)\)-bit subkey space. Therefore, the target subkeys can be recovered with a time complexity of \(O\bigl(2^{((r-12)n+3n)/2}\bigr)\). Compared with the best previously known quantum key-recovery attack under the qCPA model, our attack reduces the query complexity by a factor of \(2^{n/2}\).

\section{Quantum attack on LBlock}\label{sec4}

In this section, we study the security of LBlock under the qCPA model. We first explore the internal structure of reduced-round LBlock and construct a 9-round quantum periodic distinguisher. Based on this distinguisher, we further perform a 15-round quantum key-recovery attack. By analyzing the key scheduling of LBlock, we derive nine sets of recurrence relations among the subkeys. These relations enable us to reduce the number of guessed subkey bits in the quantum key-recovery attack, thereby shrinking the Grover search space by 24 bits.

\subsection{Description of LBlock}
LBlock is a lightweight block cipher \cite{LBlock2011}. It adopts a variant of the Feistel structure and consists of 32 rounds. The encryption process is shown in Fig.~\ref{32LBlock}. In each round, the right branch performs an 8-bit cyclic shift operation to enhance the diffusion capability and overall security of the cipher. LBlock has a 64-bit block size and an 80-bit master key. Its decryption process is the inverse of the encryption process.

\begin{figure}
    \centering
    \includegraphics[width=0.4\linewidth]{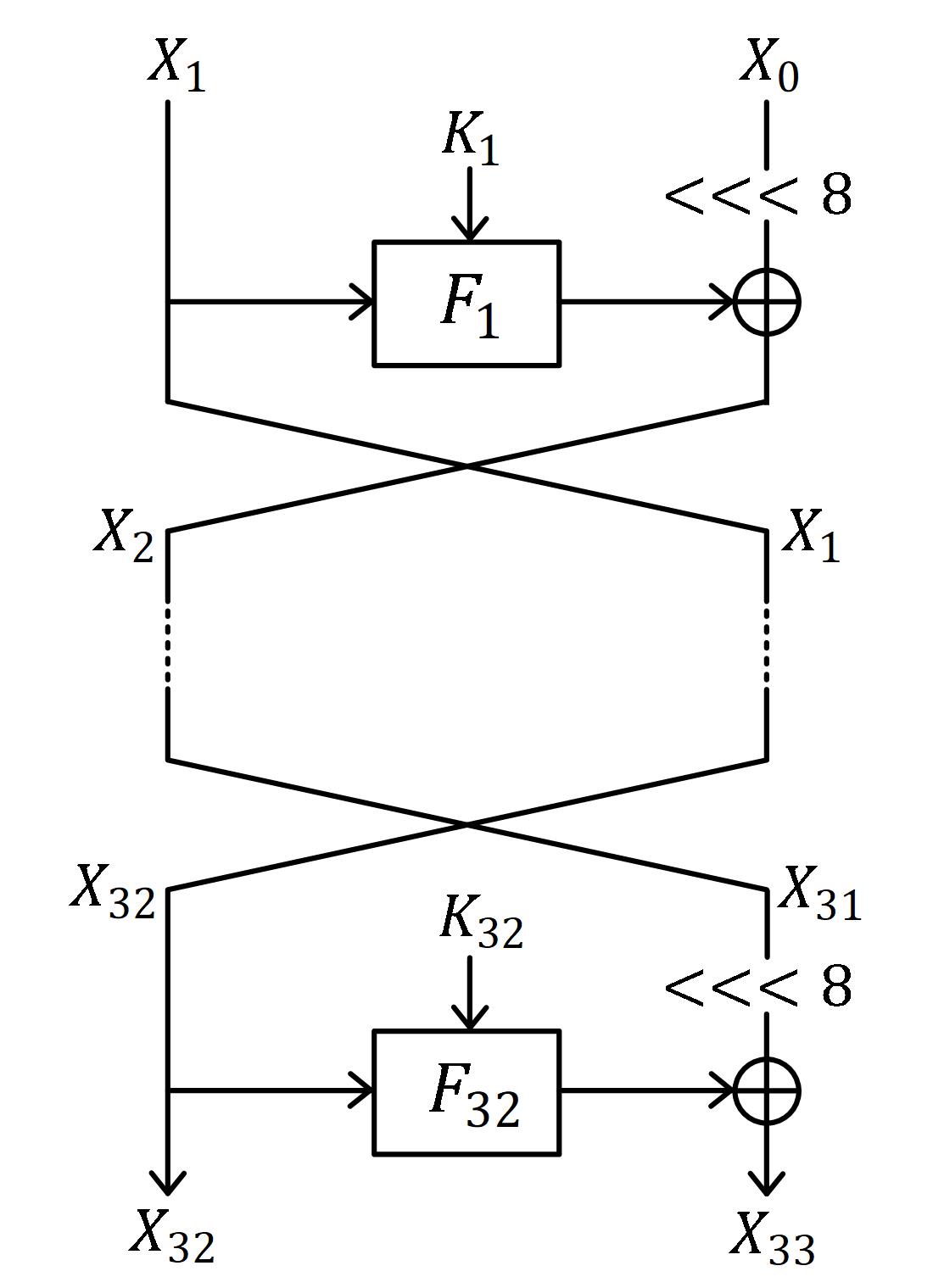}
    \caption{Encryption procedure of LBlock.}
    \label{32LBlock}
\end{figure}

LBlock comprises three components, namely the encryption algorithm, the decryption algorithm, and key scheduling. Specifically, the round function of the encryption algorithm consists of KeyAddition, Confusion Transformation (implemented using eight S-boxes), and Diffusion Transformation (implemented by the permutation $P$). The structure of the round function is shown in Fig. \ref{FunctionF}.

\begin{figure}
    \centering
    \hspace*{1cm}
    \includegraphics[width=0.8\linewidth]{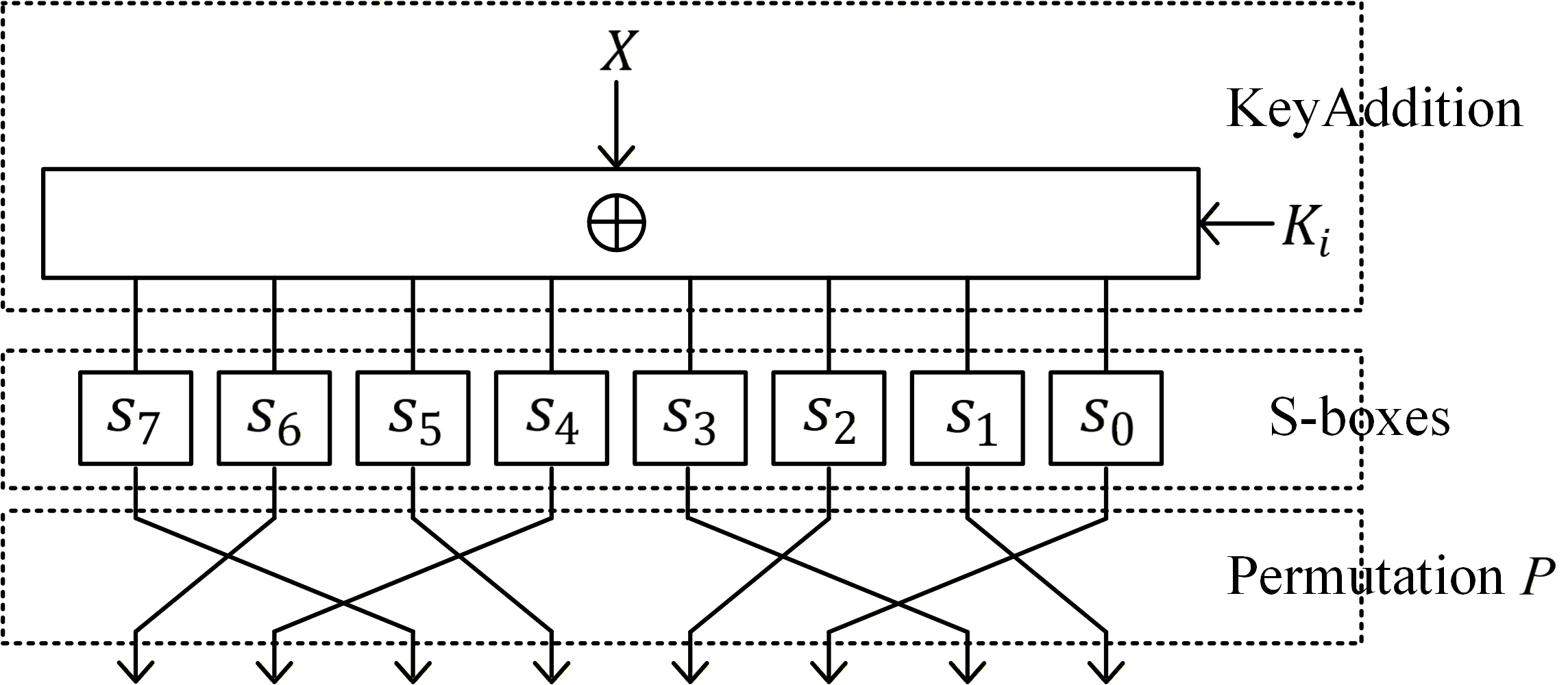}
    \caption{Round function F.}
    \label{FunctionF}
\end{figure}

In our quantum key-recovery attack, the internal structure of the round function need not be considered. Therefore, the round function is simply denoted by \(F_j\), where \(j \in \{0,1,\ldots,7\}\) is the index of the S-box.
Based on the structure of the \(i\)-th round of LBlock shown in Fig. \ref{oneLBlock}, the input-output relations among the branches are given as follows.
\[
\begin{aligned}
&x_0^i  \leftarrow F_6(x_1^{i-1}) \oplus x_{10}^{i-1},
\hspace{1.4cm}
x_1^i  \leftarrow F_4(x_3^{i-1}) \oplus x_{11}^{i-1}, \\
&x_2^i  \leftarrow F_7(x_0^{i-1}) \oplus x_{12}^{i-1},
\hspace{1.4cm}
x_3^i  \leftarrow F_5(x_2^{i-1}) \oplus x_{13}^{i-1}, \\
&x_4^i  \leftarrow F_2(x_5^{i-1}) \oplus x_{14}^{i-1},
\hspace{1.4cm}
x_5^i  \leftarrow F_0(x_7^{i-1}) \oplus x_{15}^{i-1}, \\
&x_6^i  \leftarrow F_3(x_4^{i-1}) \oplus x_8^{i-1},
\hspace{1.4cm}
x_7^i  \leftarrow F_1(x_6^{i-1}) \oplus x_9^{i-1}, \\
&x_8^i  \leftarrow x_0^{i-1},
\hspace{0.4cm}
x_9^i  \leftarrow x_1^{i-1},
\hspace{0.4cm}
x_{10}^i \leftarrow x_2^{i-1},
\hspace{0.35cm}
x_{11}^i \leftarrow x_3^{i-1}, \\
&x_{12}^i \leftarrow x_4^{i-1},
\hspace{0.28cm}
x_{13}^i \leftarrow x_5^{i-1},
\hspace{0.3cm}
x_{14}^i \leftarrow x_6^{i-1},
\hspace{0.3cm}
x_{15}^i \leftarrow x_7^{i-1}.
\end{aligned}
\]

\begin{figure}
    \centering
    \includegraphics[width=0.9\linewidth]{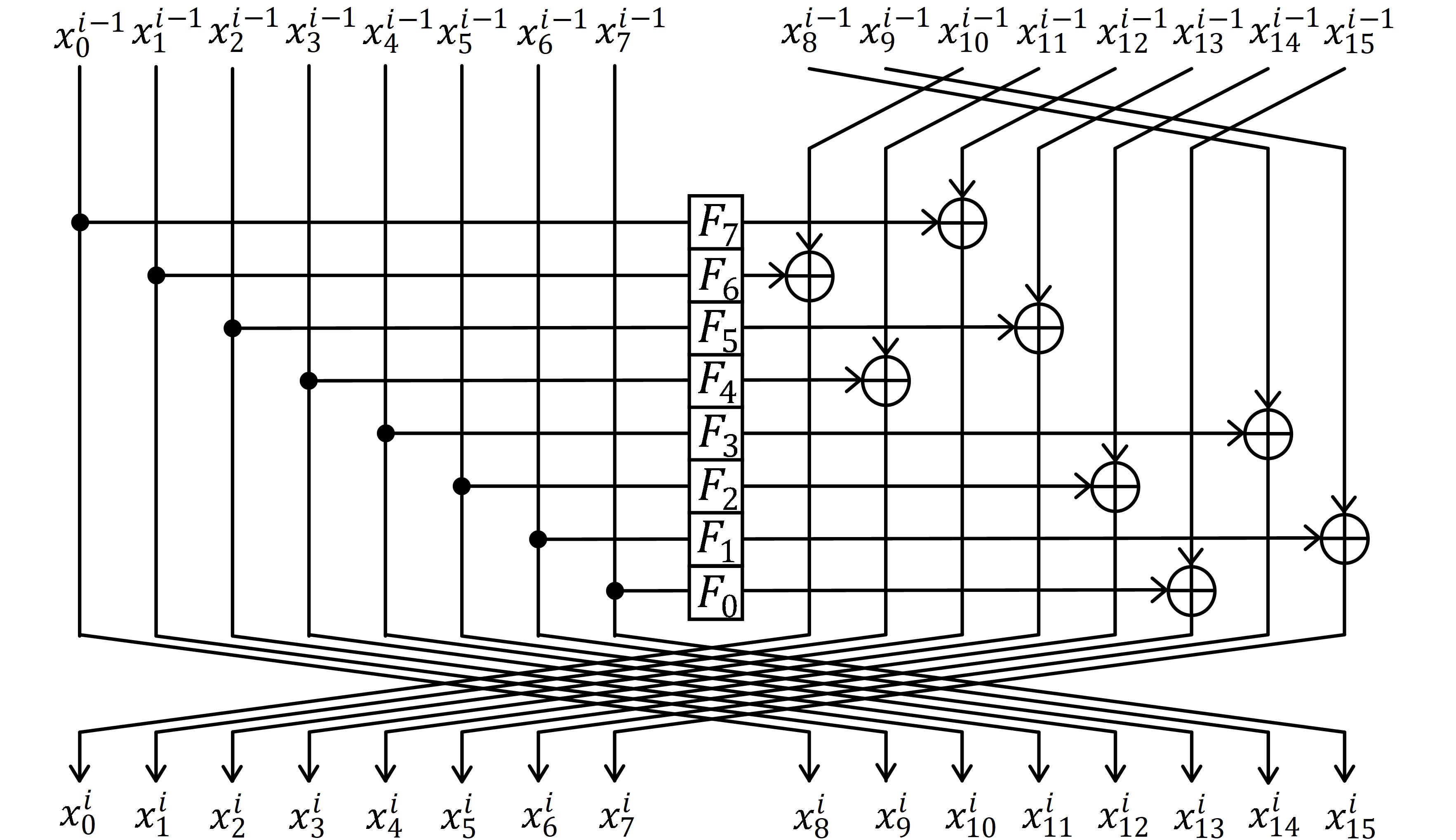}
    \caption{The $i$-th round of LBlock encryption.}
    \label{oneLBlock}
\end{figure}

The key scheduling algorithm of LBlock is inspired by the design principles of stream ciphers. It generates 32-bit round subkeys using only simple cyclic shifts and nonlinear substitutions. Specifically, the 80-bit master key is stored in a register, denoted by \( K = k_{79}k_{78}\cdots k_0 \). The leftmost 32 bits of the register are extracted as the initial round subkey \( K_1 \). For \( i = 1,2,\ldots,31 \), the key register \( K \) is updated as follows.

\begin{enumerate}
    \item Perform a 29-bit cyclic left shift on the entire register: \(K \leftarrow K \lll 29\).
    \item Apply \(s_9\) and \(s_8\) to the highest 8 bits. The specific forms of the S-boxes are given in Ref.~\cite{LBlock2011}.
    \vspace{-0.2em}
    \begin{align*}
    [k_{79}\,k_{78}\,k_{77}\,k_{76}] &\leftarrow s_9([k_{79}\,k_{78}\,k_{77}\,k_{76}]),\\
    [k_{75}\,k_{74}\,k_{73}\,k_{72}] &\leftarrow s_8([k_{75}\,k_{74}\,k_{73}\,k_{72}]).
\end{align*}

    \item Update \([k_{50}\,k_{49}\,k_{48}\,k_{47}\,k_{46}]\) by XORing it with \([i]_2\), where \([i]_2\) denotes the 5-bit binary representation of the round index \(i\).
   \vspace{-0.2em}
    \[
    [k_{50}\,k_{49}\,k_{48}\,k_{47}\,k_{46}] \leftarrow [k_{50}\,k_{49}\,k_{48}\,k_{47}\,k_{46}] \oplus [i]_2.
    \]
\vspace{-1.2em}
    \item Output the leftmost 32 bits of the current register state as the next round subkey \( K_{i+1} \).
\end{enumerate}

Due to its lightweight design and efficient hardware implementation, LBlock has been widely studied in cryptanalysis, including differential attacks \cite{Shi2024}, impossible differential attacks \cite{Wang2015}, and other types of cryptanalysis \cite{Soleimany2014,Hadipour2022,CuiY2020}.
In the quantum security analysis of LBlock, Xiang et al.~\cite{Xiang2024} proposed a quantum attack method that constructs periodic functions using probability-1 truncated differentials. Based on this technique, they constructed an 8-round quantum distinguisher. In 2025, Liu et al.~\cite{Liu2025} proposed the first fully automated SMT-based search model and obtained another 8-round quantum periodic distinguisher. In this work, we construct a 9-round quantum distinguisher for LBlock, improving upon the previous 8-round results.

\subsection{9-round quantum distinguisher for LBlock }

For LBlock, which is based on a variant of the Feistel structure, we construct a periodic function \(f^{\mathcal{O}_3}(x)\). This function can be distinguished from a random permutation by Simon algorithm in \(O(n)\) time, where \(n\) denotes the subkey length.

Under the qCPA model, let \(\alpha_0,\alpha_1\) and $x_i^0$ ($i=1,\ldots,11,13,14,15$) be arbitrary \(n\)-bit constants with \(\alpha_0 \neq \alpha_1\). Given a 9-round LBlock encryption oracle \(\mathcal{O}_3\), define the function
\[
\begin{aligned}
f^{\mathcal{O}_3} : \{0,1\}^n &\rightarrow \{0,1\}^n \\
x &\mapsto x_{10}^9(\alpha_0)\oplus x_{10}^9(\alpha_1).
\end{aligned}
\]
Here, $x_{i}^j(\alpha_b)$ denotes the output of the $(i+1)$-th branch after the $j$-th round of LBlock for the plaintext $(\alpha_b, x_1^0, \dots, x_{11}^0, x, x_{13}^0, x_{14}^0,  x_{15}^0)$.
The 9-round LBlock encryption structure is illustrated in Fig.~\ref{9LBlock}.
Based on the above periodic function, we construct a 9-round quantum distinguisher for LBlock. Its correctness is proved in Lemma~\ref{Lemma2}.

\begin{figure*}[p]
    \centering
    \includegraphics[width=0.9\textwidth]{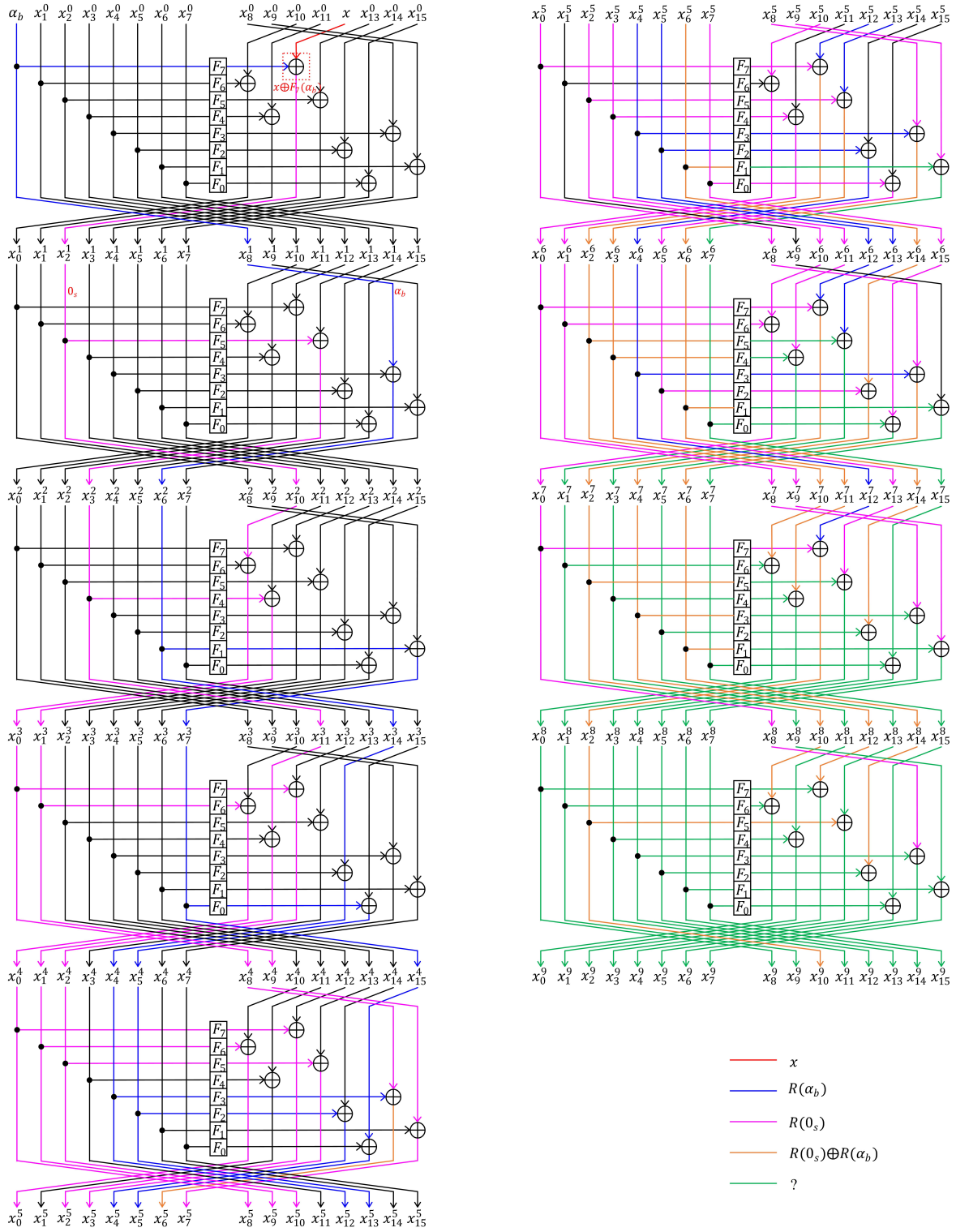}
    \caption{9-round quantum distinguisher for LBlock. Different colored lines indicate function terms with different dependencies: red lines denote functions depending only on \(x\), blue lines denote functions depending only on \(\alpha_b\), pink lines denote functions depending only on the periodic term \(0_s\), where \(0_s = x\oplus F_7(\alpha_b)\), and orange lines denote the XOR of functions depending only on \(0_s\) and \(\alpha_b\). The symbol ``?'' denotes an unknown function value without periodicity.}
    \label{9LBlock}
\end{figure*}

\begin{Lemma}\label{Lemma2}
If $\mathcal{O}_3$ is the 9-round LBlock encryption oracle, then under the qCPA model, for any \(x \in \{0,1\}^n\), the function \(f^{\mathcal{O}_3}(x)\) satisfies
\[
f^{\mathcal{O}_3}(x) = f^{\mathcal{O}_3}\bigl(x \oplus F_7(\alpha_0 ) \oplus F_7(\alpha_1)\bigr).
\]
That is, \(f^{\mathcal{O}_3}\) has period
$
s =F_7(\alpha_0 ) \oplus F_7(\alpha_1).
$
\end{Lemma}

\begin{proof}
Consider the plaintext
$
(\alpha_b,x_1^0,\ldots,x_{11}^0,x,x_{13}^0,x_{14}^0,$ $x_{15}^0),
$
where \(b\in\{0,1\}\) and \(\alpha_0\neq\alpha_1\).
As shown in Fig.~\ref{9LBlock}, the corresponding output of the 9-round LBlock encryption oracle \(\mathcal{O}_3\) is given by
\[
(x_0^9,x_1^9,\ldots,x_{15}^9)
=
\mathcal{O}_3
(\alpha_b,x_1^0,\ldots,x_{11}^0,x,x_{13}^0,x_{14}^0,x_{15}^0).
\]

In the first round, the input branch \(\alpha_b\) is processed by the function \(F_7\), and the resulting output is XORed with the branch \(x\), yielding
\[
x_2^1 = x \oplus F_7(\alpha_b).
\]

In the second round, the branch \(x_2^1\) is taken as the input of the function \(F_5\), whose output is then XORed with \(x_{13}^1\). In the following analysis, all constant terms irrelevant to the periodic structure are omitted. We then obtain
\[
x_3^2 = F_5\bigl(x \oplus F_7(\alpha_b)\bigr).
\]
Thus, both \(x_2^1\) and \(x_3^2\) belong to the class of terms depending only on \(x \oplus F_7(\alpha_b)\), represented by the pink lines in Fig.~\ref{9LBlock}.

On the other hand, the branch $x_8^1=\alpha_b$ is XORed with $F_3(x_4^1)$. Following the above convention of omitting constant terms, we write
\[
x_6^2=\alpha_b.
\]
Therefore, $x_6^2$ depends only on $\alpha_b$ and is represented by the blue lines.

Notably, in the 5th round, the output branch \(x_6^5\) can be expressed as
\[
x_6^5 = x \oplus F_7(\alpha_b) \oplus F_3(\alpha_b).
\]
This expression contains both a component related to \(x \oplus F_7(\alpha_b)\) and a component that depends solely on \(\alpha_b\).
It should be noted that whether such branches need to be tracked further depends on the target distinguisher being constructed. If the target period is determined by
\[
F_7(\alpha_b)\oplus F_3(\alpha_b),
\]
then the branch
$
x_6^5=x\oplus F_7(\alpha_b)\oplus F_3(\alpha_b)
$
must be tracked further, since it still contains the desired periodic information. Otherwise, if the distinguisher depends only on the structure \(x \oplus F_7(\alpha_b)\), then $x_6^5$ can be abstracted into the form
\[
R\bigl(x \oplus F_7(\alpha_b)\bigr) \oplus R(\alpha_b),
\]
where $R$ denotes an unspecified composition of round functions.
All branches of this form are represented by orange lines.
In this case, after applying subsequent round functions to $x_6^5$, we obtain
\[
R\bigl(R(x\oplus F_7(\alpha_b))\oplus R(\alpha_b)\bigr),
\]
which no longer explicitly preserves the target structure. Therefore, such branches do not need to be tracked further and are represented by green lines in Fig.~\ref{9LBlock}.

For simplicity, we focus on the output of the right branch of LBlock to analyze the derivation path of \(x_{10}^9(\alpha_b)\). The detailed output values of each round are given in Table~\ref{table2}. 
\begin{table*}[p]
\centering
\renewcommand{\arraystretch}{2}
\setlength{\tabcolsep}{2.5pt}
\vspace*{7.5cm}
\rotatebox{90}{
\begin{minipage}{12cm}
\centering
\caption{Right-branch outputs in each round of LBlock.}
\label{table2}
\begin{tabular}{|c|c|c|c|c|c|c|c|c|}
\hline
\diagbox{Output}{Index $i$}
& 8 & 9 & 10 & 11 & 12 & 13 & 14 & 15 \\ \hline
$x_i^1$
&
$\alpha_b$
&
$c$
&
$c$
&
$c$
&
$c$
&
$c$
&
$c$
&
$c$
\\ \hline
$x_i^2$
&
$c$
&
$c$
&
$0_s$
&
$c$
&
$c$
&
$c$
&
$c$
&
$c$
\\ \hline
$x_i^3$
&
$c$
&
$c$
&
$c$
&
$F_5(0_s)$
&
$c$
&
$c$
&
$\alpha_b$
&
$c$
\\ \hline
$x_i^4$
&
$0_s$
&
$F_4(F_5(0_s))$
&
$c$
&
$c$
&
$c$
&
$c$
&
$c$
&
$F_1(\alpha_b)$
\\ \hline
$x_i^5$
&
$F_6(F_4(F_5(0_s)))$
&
$F_5(0_s)$
&
$F_7(0_s)$
&
$c$
&
$\alpha_b$
&
$F_0(F_1(\alpha_b))$
&
$c$
&
$c$
\\ \hline
$x_i^6$
&
$F_6(F_5(0_s))$
&
$c$
&
$F_7(F_6(F_4(F_5(0_s))))$
&
$F_5(F_7(0_s))$
&
$F_2(F_0(F_1(\alpha_b)))$
&
$F_1(\alpha_b)$
&
$x\oplus F_7(\alpha_b)\oplus F_3(\alpha_b)$
&
$F_4(F_5(0_s))$
\\ \hline
$x_i^7$
&
$F_7(0_s)$
&
$F_4(F_5(F_7(0_s)))$
&
$F_7(F_6(F_5(0_s)))\oplus \alpha_b$
&
$\delta_1$
&
$F_2(F_1(\alpha_b))$
&
$F_0(F_4(F_5(0_s)))$
&
$\delta_2$
&
$-$
\\ \hline
$x_i^8$
&
$\delta_3$
&
$-$
&
$\delta_4$
&
$-$
&
$\delta_5$
&
$-$
&
$\delta_6$
&
$-$
\\ \hline
$x_i^9$
&
$-$
&
$-$
&
$\delta_7$
&
$-$
&
$-$
&
$-$
&
$-$
&
$-$
\\ \hline
\end{tabular}
\vspace{1ex}
\begin{minipage}{\textwidth}
\noindent\textbf{Note:}
\begin{align*}
0_s=& x \oplus F_7(\alpha_b)\\
\delta_1
=&
F_5\!\left(
F_7\!\left(
F_6(F_4(F_5(0_s)))
\right)\right)
\oplus
F_0(F_1(\alpha_b))
\\
\delta_2
=&
F_6(F_4(F_5(0_s)))
\oplus
F_3(F_2(F_0(F_1(\alpha_b))))
\\
\delta_3
=&
F_7(F_6(F_4(F_5(0_s))))
\oplus
F_6(F_4(F_5(F_7(0_s))))\\
\delta_4
=&
F_7(F_7(0_s))
\oplus
F_2(F_0(F_1(\alpha_b)))\\
\delta_5
=&
F_2(F_0(F_4(F_5(0_s))))
\oplus
x
\oplus
F_7(\alpha_b)
\oplus
F_3(\alpha_b)\\
\delta_6
=&
F_6(F_5(0_s))
\oplus
F_3(F_2(F_1(\alpha_b)))\\
\delta_7
=&
F_7\bigl(
F_7(
F_6(
F_4(
F_5(0_s)
)))
\oplus
F_6(
F_4(
F_5(
F_7(0_s)
)))
\bigr)
\oplus
F_2(F_1(\alpha_b)).&&
\end{align*}
\end{minipage}
\end{minipage}
}
\end{table*}
In particular, the output branch \(x_{10}^9(\alpha_b)\) can be written as
\[
\begin{aligned}
x_{10}^9(\alpha_b)=&F_7\bigl(F_7(F_6(F_4(F_5(x \oplus F_7(\alpha_b))))) \oplus \\&\quad  F_6(F_4(F_5(F_7(x \oplus F_7(\alpha_b)))))\bigr) \oplus F_2(F_1(\alpha_b)).
\end{aligned}
\]
Therefore, the function \(f^{\mathcal{O}_3}(x)\) can be expressed as
\[
\begin{aligned}
f^{\mathcal{O}_3}(x) &= x_{10}^9(\alpha_0) \oplus x_{10}^9(\alpha_1) \\
&= F_7\bigl( F_7( F_6(F_4(F_5(x \oplus F_7(\alpha_0)))) ) \oplus\\&\qquad  F_6( F_4(F_5(F_7(x \oplus F_7(\alpha_0)))) ) \bigr)  \oplus F_2(F_1(\alpha_0)) \\
&\quad \oplus F_7\bigl( F_7( F_6(F_4(F_5(x \oplus F_7(\alpha_1)))) ) \oplus\\&\qquad  F_6( F_4(F_5(F_7(x \oplus F_7(\alpha_1)))) ) \bigr)  \oplus F_2(F_1(\alpha_1)).\\
\end{aligned}
\]
The function \(f^{\mathcal{O}_3}(x)\) has a period \(s\), since it satisfies
\[
\begin{aligned}
f^{\mathcal{O}_3}&\bigl(x\oplus F_7(\alpha_0)\oplus F_7(\alpha_1)\bigr)\\
&=F_7\bigl( F_7( F_6( F_4( F_5( x \oplus F_7(\alpha_1) ) ) ) ) \oplus\\&\qquad  F_6( F_4( F_5( F_7( x \oplus F_7(\alpha_1) ) ) ) ) \bigr) \oplus F_2(F_1(\alpha_1)) \\
& \quad \oplus F_7\bigl( F_7( F_6( F_4( F_5( x \oplus F_7(\alpha_0) ) ) ) ) \oplus\\&\qquad  F_6( F_4( F_5( F_7( x \oplus F_7(\alpha_0) ) ) ) ) \bigr)
  \oplus F_2(F_1(\alpha_0))\\
&=f^{\mathcal{O}_3}(x),
\end{aligned}
\]
where $s = F_7(\alpha_0)\oplus F_7(\alpha_1)$. Thus, the lemma is proved.
\end{proof}
Through rigorous mathematical analysis, we show that $f^{\mathcal{O}_3}$ is a periodic function. Consequently, under the qCPA
model, we construct a 9-round quantum distinguisher for LBlock. By
applying Simon algorithm,  the hidden period $s$ can be efficiently recovered with $O(n)$ queries.

\subsection{Quantum key-recovery attack on LBlock}
Based on the constructed 9-round quantum distinguisher, we propose a 15-round quantum key-recovery attack on LBlock. The attack procedure consists of the following three steps.

(1) Query the 15-round LBlock encryption oracle \(\mathcal{O}_4\) with the plaintext
\((\alpha_b, x_1^0, \dots, x_{11}^0, x,x_{13}^0 \dots, x_{15}^0)\)
to obtain the corresponding ciphertext
\((x_0^{15}, x_1^{15}, \dots, x_{15}^{15})\).

(2) Apply Grover algorithm to recover the subkeys involved in the last six rounds. As shown by the orange dashed boxes in Fig.~\ref{fig3}, all corresponding round functions involve the subkeys to be guessed. Then, perform a partial decryption of the last six rounds to obtain the intermediate state
$
x_{10}^9.
$

\begin{figure}
    \centering
    \includegraphics[width=0.78\linewidth]{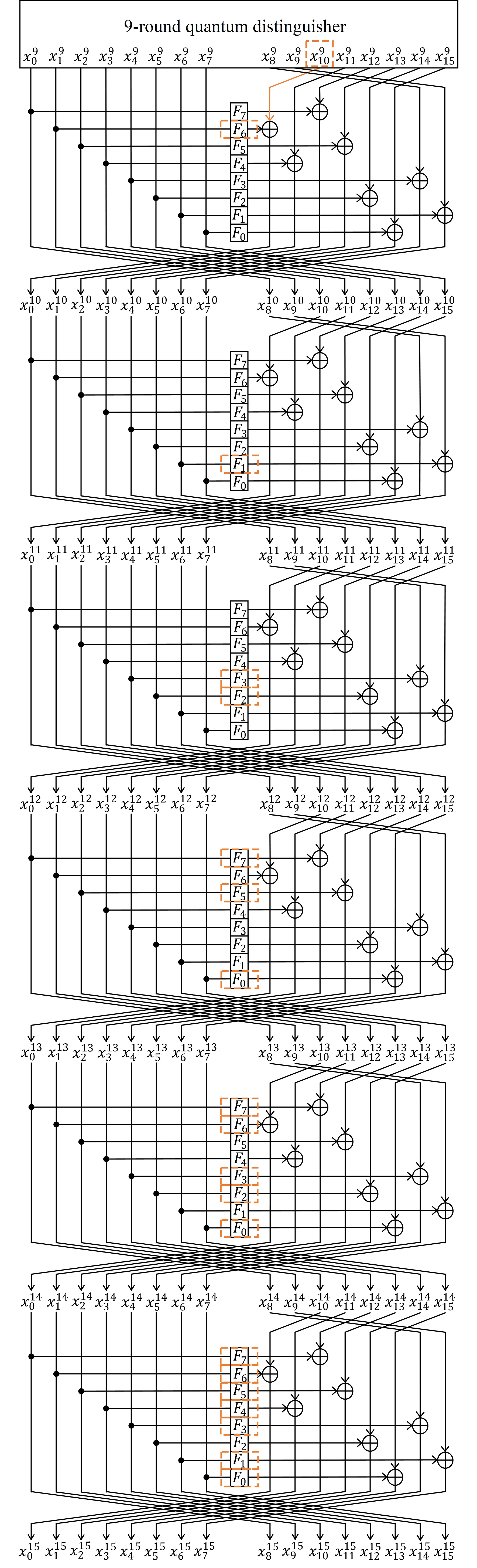}
    \caption{Quantum key-recovery attack on LBlock.}
    \label{fig3}
\end{figure}

(3) For each candidate subkey set, we apply the 9-round quantum distinguisher to verify its correctness. Only the correct subkey set can produce the periodic function described in Lemma~\ref{Lemma2}, whose hidden period can be recovered by Simon algorithm.

According to Step~(2), the attack requires guessing \(19n\) bits of subkeys using Grover algorithm, leading to a time complexity of \(O(2^{9.5n})\). More generally, for an \(r(r > 15)\)-round LBlock, the quantum key-recovery attack requires guessing \(19n + (r - 15)\times 8n\) bits of the subkeys. Consequently, the time complexity becomes \(O(2^{9.5n + (r - 15) \times 4n})\). The above complexities represent the generic attack complexities, where \(n\) denotes the subkey length.

For the LBlock lightweight block cipher, each round key consists of 32 bits, and each of the eight round functions is combined with a 4-bit subkey. By substituting \(n=4\) into the generic complexity formula, the attack described in Step~(2) requires only guessing 76 bits of subkeys using Grover algorithm, yielding a time complexity of \(O(2^{38})\).
For an \(r(r > 15)\)-round LBlock structure, the quantum key-recovery attack requires guessing \(76 + (r-15)\times 32\) bits of subkeys, resulting in a time complexity of
 \(O(2^{38 + (r - 15) \times 16})\).
The complexity of our attack is lower than that of the attack in Ref.~\cite{Liu2025} by a factor of \(2^{16}\).

\subsection{Improved quantum key-recovery attack on 15-round LBlock}

We further analyze the key scheduling of LBlock and obtain nine recurrence relations among the round subkeys. Exploiting these relations, once a subset of subkey bits is guessed, several subkey bits in other rounds can be directly derived, reducing the size of the Grover search space. Consequently, the number of subkey bits that must be guessed in the 15-round quantum key-recovery attack is reduced from 76 to 52, significantly decreasing the attack complexity.

As described in Section~4.1, the key scheduling of LBlock involves only reversible operations, namely cyclic shifts, S-box substitutions, and XORs with round constants. Hence, the key scheduling is bijective and efficiently invertible. Specifically, for \( i = 2, \ldots, 32 \), given the 80-bit key register state after the \(i\)-th key scheduling iteration, the \((i-1)\)-th round key can be recovered as follows.

\begin{enumerate}
\item Eliminate the round-constant contribution by XORing the bit segment
\([k_{50}\,k_{49}\,k_{48}\,k_{47}\,k_{46}]\)
with the round index \([i-1]_2\):
\[
[k_{50}\,k_{49}\,k_{48}\,k_{47}\,k_{46}] \oplus [i-1]_2
\rightarrow
[k_{50}\,k_{49}\,k_{48}\,k_{47}\,k_{46}].
\]
\item Apply the inverse S-boxes \(s_9^{-1}\) and \(s_8^{-1}\) to the highest 8 bits of the register.
\[
\begin{aligned}
s_9^{-1}([k_{79}\,k_{78}\,k_{77}\,k_{76}])
&\rightarrow
[k_{79}\,k_{78}\,k_{77}\,k_{76}],\\
\quad
s_8^{-1}([k_{75}\,k_{74}\,k_{73}\,k_{72}])
&\rightarrow
[k_{75}\,k_{74}\,k_{73}\,k_{72}].
\end{aligned}
\]

\item Perform a 29-bit cyclic right shift on the entire register:
$
K \ggg 29 \rightarrow K.
$

\item The register is now restored to the state corresponding to round \(i-1\). Therefore, the leftmost 32 bits of the register constitute the round key \(K_{i-1}\).
\end{enumerate}

Let \(K_i\) denote the \(i\)-th round key. The 32-bit round key \(K_i\) is partitioned into eight 4-bit subkeys as
\[
K_i=[K_i]_7\|\cdots\|[K_i]_1\|[K_i]_0,
\]
where \([K_i]_m\) denotes the \((m+1)\)-th nibble of \(K_i\) for \(m=0,1,\ldots,7\).
Furthermore, let \([K_i]_m[j]\) denote the \(j\)-th bit of the round subkey \([K_i]_m\), where \(j = 1,2,3,4\).
For brevity, \([K_i]_m[j_1,j_2,\ldots,j_t]\) denotes bits \(j_1\) through \(j_t\) of the \((m+1)\)-th nibble of the round key \(K_i\).

Using the inverse key scheduling procedure described above, the following relations among the round subkeys can be derived:
\[
\begin{cases}
[K_{15}]_7 \Rightarrow [K_{14}]_0[3,2,1],\\
[K_{15}]_1 \Rightarrow [K_{13}]_7[2,1],\\
[K_{15}]_0 \Rightarrow [K_{13}]_5[4,3],\\
[K_{14}]_7 \Rightarrow [K_{13}]_0[3,2,1],\\
[K_{15}]_{4,5} \Rightarrow [K_{12}]_3[4,3,2,1],\\
[K_{15}]_{3,4} \Rightarrow [K_{12}]_2[4,3,2,1],\\
[K_{14}]_{2,3} \Rightarrow [K_{11}]_1[4,3,2,1],\\
[K_{15}]_3 \Rightarrow [K_{10}]_6[4],\\
[K_{13}]_7 \Rightarrow [K_{10}]_6[1].
\end{cases}
\]

In summary, the invertibility of the LBlock key scheduling enables the derivation of several relations among round subkey bits. These relations allow multiple unknown subkey bits to be recovered from already guessed bits, thereby eliminating the need to include them in the Grover search space. As a result, the number of subkey bits that must be guessed in the 15-round quantum key-recovery attack is reduced from 76 to 52, which substantially improves the efficiency of the attack.

\section{Conclusion}
In this paper, we analyze the quantum security of the New Structure IV construction and the lightweight block cipher LBlock under the qCPA model. For New Structure IV, we construct an 8-round quantum distinguisher and perform a quantum key-recovery attack for the \(r(r>12)\)-round structure, achieving a better time complexity than the best existing result. For LBlock, we design a 9-round quantum distinguisher and construct a general quantum key-recovery attack. Furthermore, by exploiting the properties of its round function and key scheduling, we effectively reduce the subkey guessing space, thereby significantly improving the attack complexity. Comprehensive analysis shows that our methods achieve improvements in both the number of rounds covered by the distinguisher and the time complexity, providing new insights into the quantum security analysis of related constructions.

Future work may further extend our methods to other generalized Feistel structures or lightweight block ciphers, and design more efficient quantum distinguishers by exploiting different structural characteristics. It is also of interest to investigate the scalability and effectiveness of these methods under stronger attack models.


\ifCLASSOPTIONcompsoc
  \section*{Acknowledgments}
\else
  \section*{Acknowledgment}
\fi

This research was funded by the National Natural Science Foundation of China (Grant Nos. 62601169, 62171131, 62272056 and 62372048), Fujian Province Natural Science Foundation (Grant Nos. 2022J01186 and 2023J01533), Fujian Province Young and Middle-aged Teacher Education Research Project (Grant No. JAT231018)  and Open Foundation of State Key Laboratory of Networking and Switching Technology (Beijing University of Posts and Telecommunications) (SKLNST-2024-1-05).



%

\end{document}